\documentclass[a4paper,UKenglish,cleveref, autoref, thm-restate, runningheads]{llncs}
\usepackage{graphicx}
\usepackage{epstopdf}

\usepackage{amsthm}
\graphicspath{{./figures/}}
\usepackage{amsmath}
 \usepackage{todonotes}
 \usepackage{appendix}
 \usepackage[T1]{fontenc}
\usepackage{thm-restate,amssymb}
\usepackage{hyperref}
\usepackage{cleveref,lineno}

\newcommand{\fvs}{{\textsc{Feedback Vertex Set}}}
\newcommand{\trihit}{{\textsc{Triangle Hitting Set}}}
\newcommand{\oct}{{\textsc{Odd Cycle Transversal}}}

\theoremstyle{definition}

\DeclareUnicodeCharacter{2217}{\ensuremath{\star}}
\newbox\ProofSym
\setbox\ProofSym=\hbox{%
	\unitlength=0.18ex%
	\begin{picture}(10,10)
	\put(0,0){\framebox(9,9){}}
	\put(0,3){\framebox(6,6){}}
	\end{picture}}

\begin{document}
\title{Faster Algorithms for Cycle Hitting Problems on Disk Graphs\thanks{This work was supported by the National Research Foundation of Korea(NRF) grant 
funded by the Korea government(MSIT) (No.RS-2023-00209069)}}
%
%
\author{Shinwoo An \and
Kyungjin Cho \and
Eunjin Oh}
\authorrunning{An et al.}
%
\institute{POSTECH, Pohang, Korea\\
\email{\{shinwooan,kyungjincho,eunjin.oh\}@postech.ac.kr}}
\maketitle              
\begin{abstract}
	In this paper, we consider three hitting problems
		on a disk intersection graph: \textsc{Triangle Hitting Set},
		\textsc{Feedback Vertex Set}, and \textsc{Odd Cycle Transversal}.
		Given a disk intersection graph $G$, our goal is to compute
		a set of vertices
        hitting all triangles, all cycles, or all odd cycles, respectively.
		Our algorithms run in time
		$2^{\tilde{O}({k}^{4/5})}n^{O(1)}$, $2^{\tilde{O}({k}^{9/10})}n^{O(1)}$, and $2^{\tilde{O}({k}^{19/20})}n^{O(1)}$, respectively, where $n$ denotes the number of vertices of $G$.
		These do not require a geometric representation of a disk graph.
		If a geometric representation of a disk graph is given as input,
		we can solve these problems more efficiently. 
		In this way, we improve the  algorithms
		for those three problem by Lokshtanov et al. [SODA 2022]. 

\end{abstract}
\section{Introduction}
In this paper, we present subexponential parameterized algorithms
for the following three well-known parameterized problems on
disk graphs: \trihit, \fvs, and \oct.
Given a graph $G=(V,E)$ and an integer $k$,
these problems ask for finding a subset of $V$
of size $k$ that hits
all triangles, cycles, and odd-cycles, respectively.
On general graphs, the best-known algorithms
for \trihit, \fvs, and \oct{} run in $2.1^kn^{O(1)}$,
$2.7^kn^{O(1)}$, and $2.32^kn^{O(1)}$ time, respectively,
where $n$ denotes the number of vertices of a graph~\cite{li2020detecting,lokshtanov2014faster,wahlstrom2007algorithms}.
All these problems are NP-complete, and moreover,
no $2^{o(k)}n^{O(1)}$ algorithm exists for these problems
on general graphs
unless ETH fails~\cite{pBook}.
This motivates
the study of these problems on special graph classes
such as planar graphs and geometric intersection graphs.

Although these problems are NP-complete even for planar graphs,
much faster algorithms exist for planar graphs. More specifically, they can be solved in time $2^{O(\sqrt k)}n^{O(1)}$,
$2^{O(\sqrt k)}n^{O(1)}$, and $2^{O(\sqrt k\log k)}n^{O(1)}$, respectively, on planar graphs~\cite{demaine2005subexponential,lokshtanov2012subexponential}. Moreover, they are optimal unless ETH fails.
The $2^{O(\sqrt k)}n^{O(1)}$-time algorithms for \trihit{} and \fvs{} follow from the \emph{bidimensionality theory} of Demaine at al~\cite{demaine2005subexponential}.
A planar graph either has an ${r}\times {r}$ grid graph as a minor, or its treewidth\footnote{The definition of the treewidth can be found in Section~\ref{sec:preliminaries}} is $O({r})$. This implies that for a \textsf{YES}-instance $(G,k)$ of \fvs{}, the treewidth of $G$ is $O(\sqrt{k})$.
Then a standard dynamic programming approach gives
a $2^{O(\sqrt k)}n^{O(1)}$-time algorithm for \fvs.
This approach can be generalized for $H$-minor free graphs and bounded-genus graphs.

Recently, several subexponential-time algorithms for cycle hitting problems have been studied on \emph{geometric intersection graphs}~\cite{an2021feedback,de2020framework,deberg_et_al:LIPIcs.ISAAC.2021.22,fomin2019decomposition,fomin2019finding,fomin2012bidimensionality}.
Let $\mathcal D$ be a set of geometric objects such as disks or polygons. The intersection graph $G=(V,E)$ of $\mathcal D$ 
is the graph where
a vertex of $V$ corresponds to an element of $\mathcal D$, and two vertices of $V$ are connected by an edge in $E$
if and only if their corresponding elements in $\mathcal D$ intersect.
In particular, the intersection graph of (unit) disks
is called a \emph{(unit) disk graph},
and the intersection graph of interior-disjoint polygons
is called a \emph{map graph}.
Unlike planar graphs, map graphs and unit disk graphs
can have large cliques.
For unit disk graphs, \fvs{} can be solved in $2^{O(\sqrt{k})}n^{O(1)}$ time~\cite{an2021feedback},
and \oct{} can be solved in
$2^{O(\sqrt{k}\log k)}n^{O(1)}$ time~\cite{unitdisk-oct}. 
For map graphs, \fvs{} can be solved in
$2^{O(\sqrt{k}\log k)}n^{O(1)}$ time~\cite{fomin2019decomposition}.
These are almost tight in the sense that
no $2^{O(\sqrt{k})}n^{O(1)}$-time
algorithm exists unless ETH fails~\cite{fomin2019decomposition,fomin2012bidimensionality}.

However, until very recently,
little has been known for disk graphs, a broad class
of graphs that generalizes planar graphs and unit disk graphs.
Very recently, Lokshtanov et al.~\cite{doi:10.1137/1.9781611977073.80} presented
the first subexponential-time algorithms for cycle hitting problems on disks graphs.
The explicit running times of the algorithms are summarized
in the first row of Table~\ref{tab:results}.
On the other hand, the best-known lower bound
on the computation time for these problems
is $2^{\Omega(k^{1/2})}n^{O(1)}$ assuming ETH.
There is a huge gap between the
upper and lower bounds.

\subsubsection{Our results.}
In this paper, we make progress on the study of
subexponential-time parameterized algorithms on disk graphs
by presenting faster algorithms for the cycle hitting problems.
We say an algorithm is \emph{robust} if this algorithm does not requires a geometric representation of a disk graph. 
We present $2^{\tilde O(k^{4/5})}n^{O(1)}$-time,  $2^{\tilde O(k^{9/10})}n^{O(1)}$-time, and $2^{\tilde O(k^{19/20})}n^{O(1)}$-time robust algorithms for \trihit, \fvs, and \oct, respectively. 
 Furthermore, we present $2^{\tilde O(k^{2/3})}n^{O(1)}$-time, $2^{\tilde O(k^{7/8})}n^{O(1)}$-time, and 
 $2^{\tilde O(k^{15/16})}n^{O(1)}$-time algorithms
 for \trihit, \fvs, and \oct{} which are not robust. These results are summarized in the second and third rows of Table~\ref{tab:results}.

For \trihit, we devised a kernelization algorithm based on a crown decomposition by modifying the the algorithm in~\cite{ABUKHZAM2010524}. 
After a branching process, we obtain a set of $O((k/p)\log k)$
instances $(G',k')$ such that $G'$ has a set of $O(kp)$ triangles, which we call a \emph{core}, for a value $p$. 
Every triangle of  $G'$ 
shares at least two vertices with at least one triangle of a core. 
This allows us to remove several  vertices from $G'$ so that
the number of vertices of $G'$ is $O(kp)$. Then we can obtain a tree decomposition of small treewidth, and then we apply dynamic programming on the tree decomposition. 

For \fvs{} and \oct, 
we give an improved analysis of the algorithms in~\cite{doi:10.1137/1.9781611977073.80} by presenting
improved bounds on one of 
the two main combinatorial results 
presented in~\cite{doi:10.1137/1.9781611977073.80} 
concerning the treewidth of disk graphs (Theorem~2.1). 
More specifically, Theorem~1 of~\cite{doi:10.1137/1.9781611977073.80} says that 
for any subset $M$ of $V$ such that $N(v)\cap M \neq N(u)\cap M$ 
for any two vertices $u$ and $v$ not in $M$, the size of $U$ is $O(|M|\cdot p^6)$, where $p$ is the ply of the disks represented by the vertices of $G$. 
We improve an improved bound of $O(|M|\cdot p^2)$.  
To obtain an improved bound, we classify the vertices of $V-M$ 
into three classes in a more sophisticated way, and then use the concept of the additively weighted higher-order Voronoi diagram. 
We believe the additively weighted higher-order Voronoi diagram will be useful in designing optimal algorithms for \trihit, \fvs, and \oct. Indeed, the Voronoi diagram (or  Delaunay triangulation) is a main tool for obtaining an optimal algorithm for \fvs{} on unit disk graphs~\cite{an2021feedback}.

\begin{table}[t]
    \centering
    \begin{tabular}{|  c c c | c | c|}
    \hline
     \textsc{Triangle Hitting} & \textsc{FVS} & \textsc{OCT} & robust? & \\ 
    \hline\hline
        $2^{ O(k^{9/10}\log k)}n^{O(1)}$ & $2^{ O(k^{13/14}\log k)}n^{O(1)}$ & $2^{ O(k^{27/28}\log k)}n^{O(1)}$ & yes
        & \cite{doi:10.1137/1.9781611977073.80} \\
        [0.5ex]
                 $2^{  O(k^{4/5}\log k)}n^{O(1)}$ & $2^{ O(k^{9/10}\log k)}n^{O(1)}$ &   $2^{ O(k^{19/20}\log k)}n^{O(1)}$
        & yes & this paper \\
         $2^{ O(k^{2/3}\log k)}n^{O(1)}$ & $2^{  O(k^{7/8}\log k)}n^{O(1)}$ &  $2^{ O(k^{15/16}\log k)}n^{O(1)}$ & no 
        & this paper\\
        \hline   
    \end{tabular}
    \vspace{2mm}
    \caption{Comparison of the running times of our algorithms and the algorithm by Lokshtanov et al.~\cite{doi:10.1137/1.9781611977073.80}. 
    All algorithms for \oct{} are randomized. 
    }
    \label{tab:results}
\end{table}

\section{Preliminaries}   \label{sec:preliminaries}
For a graph $G$, we let $V(G)$ and $E(G)$ be the sets of 
vertices and edges of $G$, respectively. For a subset $U$ of $G$, we use $G[U]$ to
denote the subgraph of $G$ induced by $U$. 
Also, for a subset $U\subset V$ of $V$, we simply denote $G[V\setminus U]$ by $G-U$. 
For a vertex $v$ of $G$, let $N(v)$ be the set of 
the vertices of $G$ adjacent to $v$. We call it the \emph{neighborhood} of $v$.

A triangle of $G$ is a cycle of $G$ consisting of three vertices.
We denote a triangle consisting of three vertices $x$, $y$ and $z$ by $\{x,y,z\}$. We sometimes consider it as the set $\{x,y,z\}$ of vertices. For instance, the union of triangles
is the set of all vertices of the triangles.
A subset $F$ of $V$ is called 
a \emph{triangle hitting set} of $G$ if 
 $G[V\setminus F]$ has no triangle.
Also, $F$ is called 
a \emph{feedback vertex set} if 
$G[V\setminus F]$ has no cycle (and thus it is forest). 
Finally, $F$ is called a \emph{odd cycle transversal} if $G[V\setminus F]$ has no odd cycle.
In other words, a triangle hitting set, a feedback vertex set,
and a odd cycle transversal
hit all triangles, cycles and odd cycles, respectively. 
Notice that a feedback vertex set of $G$ 
is also a triangle hitting set.

\subsubsection{Disk graphs.}
Let $G=(V,E)$ be a disk graph defined by a set $\mathcal D$
of disks. 
In this case, we say that $\mathcal D$ is a \emph{geometric representation} of $G$. 
For a vertex $v$ of $G$, we let $D(v)$ denote the disk of $\mathcal D$ represented by $v$. 
The \emph{ply} of $\mathcal D$ 
is defined as the maximum number of disks of $\mathcal D$ containing a common point. Note that the disk graph of a set of disks of ply $p$ 
has a clique of size at least $p$. 
If it is clear from the context, we say that the ply of $G$ is $p$.
The \emph{arrangement} of $\mathcal D$ is the subdivision of the plane
formed by the boundaries of the disks of $\mathcal D$
that consists of vertices, edges and faces.
The \emph{arrangement graph} of $G$, denoted by $\mathcal A(G)$, is the plane graph where every
face of the arrangement of $\mathcal D$ contained in at least one disk is represented by a vertex, and vertices are adjacent if the faces they represent share an edge.
Given a disk graph $G$, it is NP-hard to compute its geometric
representation~\cite{breu1998unit}. 
We say an algorithm is \emph{robust} if this algorithm does not requires a geometric representation of a disk graph. 

\subsubsection{Tree decomposition.}
A \emph{tree decomposition} of an undirected graph $G=(V,E)$
is defined as a pair $(T,\beta)$, where $T$ is a tree and $\beta$ is a mapping from nodes of $T$ to subsets of $V$ (called bags) with the following conditions. Let $\mathcal {B}:=\{\beta(t) : t\in V(T)\}$ be the set of bags of $T$.
	\begin{itemize}
		\item  For $\forall u\in V$, there is at least one bag in $\mathcal {B}$ which contains $u$.
		\item For $\forall (u,v)\in E$, there is at least one bag in $\mathcal {B}$ which contains both $u$ and $v$.
		\item For $\forall u\in V$, the nodes of $T$ containing $u$ in their bags are connected in $T$.
	\end{itemize}
	The $\emph{width}$ of a tree decomposition is defined as the size of its largest bag minus one, and the $\emph{treewidth}$ of $G$ is
	the minimum width of a tree decomposition of $G$.
	The treewidth of a disk graph $G$ is $O(p\sqrt{|V(G)|})$, where $p$ is the ply
	of $G$~\cite{doi:10.1137/1.9781611977073.80}.

\subsubsection{Weighted Treewidth of a Disk Graph}
\label{sec:weighted-tw}
One of the ingredients of our algorithms is to partition the vertices
in each bag of a tree decomposition of $G$ into cliques.
This technique was already used in~\cite{de2020framework}
to design ETH-tight (non-parameterized) algorithms for various problems on geometric intersection graphs.
To use this observation, they introduced
the notion of \emph{clique-weighted treewidth}.

We fix a geometric representation of $G$. 
For a subset $U$ of $V$, 
we say the partition $\{C_1,\ldots C_\ell\}$ of $U$ is a \emph{clique partition} if 
all disks represented by the vertices of $C_i$
contain a common point in the plane 
for each index $i$.
We define the \emph{clique-weight}
of $U$ as the minimum of  $\sum_{1\leq i\leq \ell}(\log |C_i|+1)$  over all
 clique partitions $\{ C_1,\ldots C_\ell\}$ of $U$.
Since a triangle hitting set contains
all except for at most two vertices of $C_i$,
the number of 
intersections between $U$ and a triangle hitting set is $\prod_{1\leq i\leq t} O(|C_i|^2)=2^{O(w)}$, where $w$ is the clique-weight of $U$.

A \emph{balanced separator} $S_\text{sep}$ of $G$ is a subset of vertices
such that each connected component of $G-S_\text{sep}$
has complexity $c|V|$ for a constant $c<1$.
De berg et al.~\cite{de2020framework} showed that
a geometric intersection graph of fat objects
has a balanced separator of small clique-weight.
\begin{lemma}[\cite{de2020framework}]\label{lem:deberg-separator}
 A disk graph $G$ has a balanced separator $F_\text{sep}$ and a clique partition of $F_\text{sep}$ with weight $O(\sqrt n)$. Moreover, if a geometric representation of $G$ is given, an $F_\text{sep}$ and its clique partition  can be computed in polynomial time.
\end{lemma}

For a tree decomposition $(T,\beta)$ of a graph, we define the clique-weighted width of $(T,\beta)$
as the maximum weight of a bag of the decomposition.
Then the clique-weighted treewidth of a graph is defined as
the minimum clique-weighted width of a tree decomposition of the graph.
In the following, we simply use the term \emph{weight} (and \emph{weighted width}) instead of clique-weight (and \emph{clique-weighted width}), if it is clear from the context.
Then Lemma~\ref{lem:deberg-separator} implies the following lemma.

\begin{lemma}\label{lem:tw-without-ply}
 The weighted treewidth of a disk graph $G$ is $O(\sqrt n)$. Given a geometric representation of $G$, we can compute a tree decomposition of $G$ with weighted treewidth $O(\sqrt n)$ in
 polynomial time.
\end{lemma}
\begin{proof}
    We construct a tree decomposition $(T,\beta)$ recursively for a disk graph $G$ with its geometric representation.
    Our construction satisfies that a subtree of $T$ rooted at a node $t$ together with 
    the bags corresponding its nodes 
    forms a tree decomposition of a subgraph of $G$, say $G_t$. 
    
    Our construction starts from the root node $r$ of $T$ with $G_r=G$. For a current node $t$ of $T$, we set $\beta(t)=V(G_t)$ if the number of vertices of $V(G_t)$ is $O(\sqrt {n})$. 
    Otherwise, we find a balanced separator $F_\text{sep}$ of $G_t$ of size $O(\sqrt{|V(G_t)|})$ and its clique partition with weight $O(\sqrt {n})$ in polynomial time by Lemma~\ref{lem:deberg-separator}. 
    For each connected component $V_i$ of $G_t-F_\text{sep}$, we 
    recursively construct a tree decomposition $(T_i,\beta_i)$ of $G[V_i]$. Then we connect the root of $T_i$ to $t$ so that it becomes a child of $t$. 
    In addition to this, we add the vertices of $F_\text{sep}$ to the bag of every node in the subtree rooted at $t$ including $t$ itself. 
    There is no edge between two different components $V_i$ and $V_j$. Thus, $(T,\beta)$ is a tree decomposition of $T$.  Furthermore, its weight treewidth is $O(\sqrt n)$ because $F_\text{sep}$ is a balanced separator. 
\end{proof}

\subsubsection{Bounded Ply.}
If the ply of a geometric representation $\mathcal D$ of $G$ is bounded,
we can obtain a better bound on the weighted treewidth of $G$. 

\begin{lemma}\label{lem:tw-arrangment}
    The weighted treewidth of $G$ is at most $\log p$ times the treewidth of $\mathcal A(G)$. Furthermore, the treewidth of $G$ is at most $p$ times the treewidth of $\mathcal A(G)$.
\end{lemma}
\begin{proof}
    Let $(T,\beta)$ be a tree decomposition of $\mathcal A(G)$.
    We can construct a tree decomposition $(T, \beta')$ of $G$ as follows.
    Each vertex of $\mathcal A(G)$ corresponds to a face of the arrangement of $\mathcal D$. Moreover,
    the disks of $\mathcal D$ containing a common face of the arrangement
    forms a clique in $G$.
    For each node $t\in T$, let $\beta'(t)$ be the set
    of the vertices corresponding to disks of $\mathcal D$  containing faces corresponding to
    the vertices of $\beta(t)$.
    It is not difficult to see that $(T,\beta')$ is a tree decomposition of $G$.
    Then the weighted width (and treewidth) of $\beta'(t)$ is $\log p$ (and $p$) times the width of $\beta(t)$. 
    This implies that the weighted treewidth (and treewidth) of $G$ is at most $\log p$ (and $p$) times the treewidth of $\mathcal A(G)$.
\end{proof}

Given a tree decomposition of $G$ with weighted width $w$, \fvs{} can be solved in $2^{O(w)}n$ time
using a rank-based approach as observed in~\cite{de2020framework}.\footnote{De berg et al.~\cite{de2020framework} showed that \fvs{} and several problems can be solved in $2^{O(\sqrt{n})}$ time for similarly sized objects. Although the other algorithms cannot be extended to an intersection graph of fat objects, the algorithm for \fvs{} can be extended to an intersection graph of fat object. We briefly show how to solve \fvs{} in Lemma~\ref{lem:dp-fvs}.}
Similarly, \trihit{} and \oct{}
can be solved in $2^{O(w)}n$ and $2^{O(w\log w)}n$ time, 
respectively. They can be obtained by slightly 
modifying standard dynamic programming algorithms 
for those problems as observed in~\cite{de2020framework}. We briefly show how to do this in Section~\ref{sec:dp}.

\subsubsection{Higher-Order Voronoi Diagram.}
   To analyze the treewidth of $G$ of bounded ply, we use the concept of the \emph{higher-order Voronoi diagram}. Given a set of $n$ weighted sites (points),  
    its order-$k$ Voronoi diagram is defined as
    the subdivision of $\mathbb{R}^2$ into maximal regions
    such that all points within a given region have the same $k$ nearest sites. 
    Here, the distance between a site $s$ with weight $w$ and a point $x$ in the plane 
    is defined as $d(s,x) - w$, where $d(s,x)$ is the Euclidean distance between $s$ and $x$. 
    \begin{lemma}[Theorem 4 in~\cite{rosenberger1991order}]\label{lem:complexity}
    The complexity of the addictively weighted order-$k$ Voronoi diagram of $m$ point sites in the plane is $O(mk)$.
    \end{lemma}


\section{Triangle Hitting Set}\label{sec:trihit}
In this section, we present a robust 
$2^{O(k^{4/5}\log k)}n^{O(1)}$-time algorithm for \trihit{}, which 
improves the $2^{O(k^{9/10}\log k)}n^{O(1)}$-time algorithm of~\cite{doi:10.1137/1.9781611977073.80}.


\subsection{Two-Step Branching Process}\label{sec:tri-branching}
We first apply a branching process as follows
to obtain $2^{O((k/p)\log k)}$ instances
one of which is  a \textsf{YES}-instance $(G',k')$ of \trihit{} where every clique of $G'$ has size $O(p)$.
For a clique $C$ of $G$ and a triangle hitting set $F$,
all except for at most two vertices of $C$ are contained in  $F$.
In particular, if $G$ has a triangle hitting set of size $k$, any clique of $G$
has size at most $k+2$.
For a clique of size at least $p$, we branch on which vertices of the clique are not included in a triangle hitting set. After this, the solution size $k$ decrease by at least $p-2$.
We repeat this until every clique has size $O(p)$.
In the resulting branching tree, every node has at most  $O(k^2)$ children since any clique of $G$ has size at most $k+2$.
And the branching tree has height $O(k/p)$.
In this way,
we can obtain $2^{O((k/p)\log k)}$ instances of \trihit{}
one of which is  a \textsf{YES}-instance $(G',k')$ of \trihit{} where $G'$ has a geometric representation of ply at most $p$. 
Moreover, $G'$ is an induced subgraph of $G$, and $k'\leq k$.
Using the EPTAS for computing a maximum clique in a disk graph~\cite{bonamy2021eptas}, we can complete the branching step in $2^{O((k/p)\log k)}n^{O(1)}$ time. 
Note that the algorithm in~\cite{bonamy2021eptas} does not require a geometric representation of a graph.

\begin{figure}[t]
    \centering
    \includegraphics[width=0.9\textwidth]{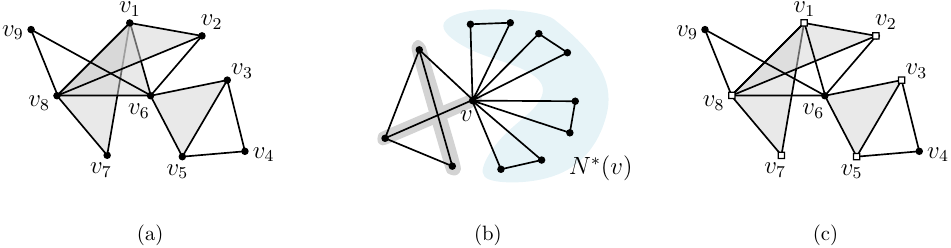}
    \caption{\small (a)
    The four gray triangles,
    $v_1v_2v_8, v_1v_6v_8, v_1v_7v_8$ and $v_3v_5v_6$,
    form a core.
    A triangle of $G$ not in the core shares two vertices with at least one gray triangle. (b) The marked edges are colored gray.
    Then $N^*(v)$ has a matching of size four.
    (c) The vertices in the initial set of $F$ are $v_1,v_2,v_3,v_5,v_7,v_8$. We add two three gray triangles $v_1v_2v_8, v_1v_7v_8, v_3v_5v_6$ to $W$. 
    Note that $W$ is not yet a core because of $v_6v_8v_9$.
    }
    \label{fig:core}
\end{figure}

For each instance $(G,k)$ obtained from the first branching,
we apply the second branching process to obtain
$2^{O(k/p)}$ instances
one of which is a \textsf{YES}-instance having
a \emph{core} of size $O(pk)$.
A set $W$ of triangles of $G$ is called a \emph{core}
if
for a triangle of $G$, a triangle of $W$ shares at least two vertices
with the triangle.
See Fig.~\ref{fig:core}(a).
During the branching process, we mark an edge
to remember that one of its endpoints must be added to a triangle hitting set.
The marking process has an invariant that no two marked edges share a common endpoint, and thus the number of
marked edges is at most $k$.
The marks will be considered in the dynamic programming procedure. Initially, all edges are unmarked.

Let $v$ be a vertex such that $N^*(v)$ has a matching of size
at least $p$, where $N^*(v)$ be the set of neighbors of $v$ not incident to any marked edge. See Fig.~\ref{fig:core}(b).
In this case, a triangle hitting set contains either $v$
or
at least one endpoint of each edge in the matching.
We branch on whether or not $v$ is added to a triangle hitting set.
For the first case, we remove $v$ and its adjacent edges, and decrease $k$ by one.
For the second case, we know that at least
one endpoint of each edge in the matching must be contained in a triangle hitting set.
However, we do not make a decision at this point.
Instead, we simply mark all such edges.

\begin{lemma}\label{lem:branching}
    The total number of instances from the
    two-step branching process is
    $2^{O((k/p)\log k)}$. Moreover, the branching process runs in $2^{O((k/p)\log k)}n^{O(1)}$ time.
\end{lemma}
\begin{proof}
    Recall that the first branching step produces $2^{O((k/p)\log k)}$ instances.
    Let $(G',k')$ be an instance of \trihit{} we produce
    during the second branching process.
    Let $N(k',m')$ be the number of instances produced by $(G',k')$ where $G'$ has $m'$ marked edges.
    By construction, we produce two instances: one has
    parameter $k'-1$, and one has at least $m'+p$ marked edges.
    Thus we have
    \[N(k', m')\leq N(k'-1, m') + N(k', m'+p).\]
    Since $N(0,\cdot)=0, N(\cdot,k)=0$, we have
    $N(k,0)=2^{O(k/p)}$. Therefore, the total number of instances obtained from the two-step branching process
    is $2^{O((k/p)\log k)}$. Moreover, since our algorithm runs in  polynomial time per instance, the total running time is $2^{O((k/p)\log k)}n^{O(1)}$.
\end{proof}

This branching process was already used in~\cite{doi:10.1137/1.9781611977073.80}. 
The following lemma is a key for our improvement over~\cite{doi:10.1137/1.9781611977073.80}.  
\begin{lemma}\label{lem:core}
    Let $(G,k)$ be a \textsf{YES}-instance obtained from the two-step branching. Then $G$ has a core of size $O(pk)$, and we can compute one in polynomial time.
\end{lemma}
\begin{proof}
    We construct a core $W$ of $G$ as follows. 
    Let $F_0$ be the the union of
    a triangle hitting set of size at most $3k$
    and the set of all endpoints of the marked edges,
    which can be computed in polynomial time.\footnote{
    We can find a hitting set of size at most $3k$ as follows: Find a triangle, and add all its vertices to a triangle hitting set. Then remove all its vertices from the graph.}
    Note that, the size of $F_0$ is at most $O(k)$. 
    Then let $W$ be the set of triangles constructed as follows:
    for each edge $xy$ of $G[F_0]$, we add
    an arbitrary triangle of $G$ formed by $x,y$ and $v\in V\setminus F_0$ to $W$, if it exists.
    The
    number of triangles in $W$ is $O(pk)$
    by Lemma~\ref{lem:technical-core}.

    At this moment, $W$ is not necessarily a core. See Fig.~\ref{fig:core}(c).
    Thus we add several triangles further to $W$ to compute a core of $G$.
    By the branching process, for every vertex $v$ of $F_0$,
    $N(v)\setminus F_0$ has a maximum matching of size
    at most $p$. We add the triangles formed by $v$ and the edges of the maximum matching to $W$ for every vertex $v$ of $F_0$. Note that we do not update $F_0$ during this phase, and thus triangles added to $W$ due to two different vertices might intersect.
    This algorithm clearly runs in polynomial time.
    Moreover, since the size of $F_0$ is $O(k)$,
    we add at most $O(pk)$ triangles to $W$, and thus
    the size of $W$ is $O(pk)$.

    We claim that $W$ is a core of $G$.
    Let $\{x,y,z\}$ be a triangle of $G$ not in $W$.
    Since $F_0$ contains a triangle hitting set of size at most $3k$, it must contain at least one of $x, y$ and $z$.
    If at least two of them, say $x$ and $y$, are contained in $F_0$,
    then there exists a triangle having edge $xy$ in $W$ by  construction. Thus, $\{x,y,z\}$ shares two vertices with
    such a triangle.
    The remaining case is that exactly one of them, say $x$,
    is contained in $F_0$.
    In this case,
    we have considered $x$ and
    a maximum matching of $N(x)\setminus F_0$.
    The only reason why $\{x,y,z\}$ is not added to $W$
    is that the edge $yz$ is adjacent to another edge $y'z'$ for
    some other triangle $\{x,y',z'\}$. Then $\{x,y',z'\}$ is
    added to $W$. Note that $\{x,y,z\}$ and
    $\{x,y',z'\}$ share at least two vertices, and thus  $\{x,y,z\}$ satisfies the condition for $W$ being a core. 
\end{proof}
\begin{lemma}\label{lem:technical-core}
    For a subset $F$ of $V$ of size $O(k)$,
    $G[F]$ has $O(pk)$ edges.
\end{lemma}
\begin{proof}
    We use a charging scheme. For each edge $uv$ of $G[F]$,
    we charge it to the smaller disk among the disks
    represented by $u$ and $v$.
    Then each disk $D$ is charged by $O(p)$ edges.
    This is because $D$ is intersected by $O(p)$ disks
    larger than $D$.
    To see this, recall that the ply of $\mathcal D$ is at most $p$ due to the first branching process.
    Since the size of $F$ is $O(k)$,
    the number of edges of $G[F]$ is $O(pk)$.
\end{proof}


The following observation will be used in the correctness proof of the kernelization step
in Section~\ref{sec:kernel}. The observation holds because $G'$ does not have any triangle not appearing in $G$.
\begin{observation}\label{obs:core}
    Let $G'$ be an induced subgraph of $G$ such that $V(G')$
    contains all vertices of a core $W$ of $G$.
    Then $W$ is a core of $G'$.
 \end{observation}

For each instance $(G,k)$ obtained from the branching process,
we apply the cleaning step that removes all vertices not hitting any triangle of $G$ from $G$.
Specifically, we remove a vertex of degree less than two. Also, we remove a vertex whose neighbors are independent. Whenever we remove a vertex,
we also remove its adjacent edges.

\subsection{Kernelization Using Crown Decomposition}\label{sec:kernel}
Let $(G,k)$ be a \textsf{YES}-instance we obtained from the
branching process.
We show that if the number of vertices of $G$
contained in the triangles of $G$ is at least $pk$,  we can
construct a \emph{crown} for the triangles of $G$.
Then using this, we can produce a \textsf{YES}-instance $(G',k)$ of \trihit{} in polynomial time where $G'$
is a proper induced subgraph of $G$.
Thus by repeatedly applying this process (at most $n^2$ times)
and then by removing all vertices not contained in any triangle of $G$,
we can obtain a \textsf{YES}-instance $(G',k)$ where $G'$ is a disk
graph of complexity $O(pk)$.

\begin{figure}[t]
    \centering
    \includegraphics[width=0.5\textwidth]{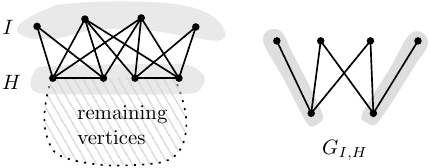}
    \caption{
    \small 
    The edges in the matching $M$ in $G_{I,H}$
    are colored gray.}
    \label{fig:crown}
\end{figure}

We first define a crown decomposition of a graph for triangles. For illustration, see Fig.~\ref{fig:crown}.
A crown decomposition of a graph
was initially introduced to construct
a linear kernel for \textsc{Vertex Cover}. Later, Abu-Khzam~\cite{ABUKHZAM2010524}
generalized this concept to hypergraphs.
Using this, he showed that a triangle hitting set admits
a quadratic kernel for a general graph.
In our case, we will show that the size of a kernel is indeed
$O(pk)$ due to the branching process. More specifically,
it is due to the existence of a core of size $O(pk)$.

\begin{definition}[{\cite{ABUKHZAM2010524}}]\label{def:crown}
A \emph{crown} for the triangles of $G$
is a triple $(I, H, M)$ with a subset $I$ of $V(G)$, a subset $H$ of $E(G)$, and a matching $M$ of $G_{I,H}$ s.t. 
\begin{itemize}
    \item no two vertices of $I$ are contained in the same triangle of $G$, 
    \item each edge of $H$ forms a triangle with some vertex of $I$, and
    \item every edge (vertex in $G_{I,H}$) of $H$ is matched under $M$, 
\end{itemize}
where $G_{I,H}$ is the bipartite graph with vertex set $I\cup H$ such that
$x\in I$ and $uv\in H$ are connected by an edge if and only if $u,v$ and $x$ form a triangle. 
\end{definition}
\medskip

If a crown $(I,H,M)$ for the triangles of $G$ exists,
we can remove all vertices in $I$, but instead, we mark all edges of $H$. Then the resulting graph also has a triangle hitting set of size at most $k$~\cite{ABUKHZAM2010524}. 
\begin{lemma}[{\cite[Lemma 2]{ABUKHZAM2010524}}]\label{lem:kernel}
    Let $(G=(V,E),k)$ be a \textsf{YES}-instance of \trihit,  and $(I,H,M)$
    is a crown for the triangles of $G$.
    Then the subgraph of $G$ obtained by removing all vertices of $I$ and by marking all edges of $H$
    has a trianlge hitting set of size at most $k$.
\end{lemma}
\begin{proof}
    To make the paper self-contained, we give a proof of this lemma, but one
    can find this proof also in~\cite{ABUKHZAM2010524}.
    Let $F$ be a triangle hitting set of $G$ of size $k$.
    We construct another hitting set $F'$ from $F$
    not containing any vertex of $I$ as follows.
    We initialize $F'$ as $F$. If
    there is an edge $uv$ of $H$
    $u\notin F$ and $v\notin F$, then $x$ must be contained in $F'$, where the edge $\{x, uv\}$ of $G_{I,H}$ is in $M$.
    Then we replace $x$ with $u$. We repeat this until at least one vertex of all edges of $H$ is contained in $F'$.
    Then $|F'|\leq |F| = k$.

    We claim that $F'$ is also a triangle hitting set.
    A triangle not intersecting $I$ is hit by $F'$ by construction.
    Thus consider a triangle $\{x,u,v\}$ intersecting $I$.
    Since no two vertices of $I$ are contained in the same triangle,
    exactly one of $x,u$ and $v$, say $x$, is contained in $I$.
    Since every edge of $H$ is matched under $M$, $uv\in H$.
    Since $F'$ contains either $u$ or $v$, $\{x,u,v\}$ is hit by $F'$.
    Therefore, $F'$ is a triangle hitting set of size at most $k$.
\end{proof}

Now we show that $(G,k)$ has a crown for its triangles if  $c pk$ vertices are contained in the triangles of $G$,
where $c$ is a sufficiently large constant.
Let $W$ be a core of $G$ of size $O(pk)$,
which exists due to Lemma~\ref{lem:core} and Observation~\ref{obs:core}.
Let $I$ be the set of vertices of $G$ not contained in any triangle of $W$ but contained in some triangle of $G$.
If $cpk$ vertices are contained in the triangles of $G$, the size of $I$ is $c'pk$ for a sufficiently large constant $c'$  since $|W|=O(pk)$.
Then let $H$ be the set of all edges of $G$
which form triangles together with the vertices of $I$.
Note that for every edge of $H$, its endpoints are contained
in the same triangle of $W$
by the definition of the core.
Thus the size of $H$ is at most $3\cdot |W|=O(pk)$.
Therefore,
if more than  $cpk$ vertices are contained in the triangles of $G$,  $|I| > |H|$.

\begin{lemma}\label{lem:crown}
     If $|I|>|H|$, there are two subsets $H'\subseteq H$ and $I'\subseteq I$ such that
     $(I',H',M')$ is a crown for the triangles of $G$.
\end{lemma}
\begin{proof}
    Let $G_{I,H}$ be the bipartite graph defined in Definition~\ref{def:crown}.
    Notice that $I$ is an independent set in $G_{I,H}$
    since $G_{I,H}$ is bipartite.
    Let $X$ be a minimum vertex cover of $G_{I,H}$.
    Also, let $M$ be a matching of size $|X|$ of $G_{I,H}$
    such that every edge in a matching is incident to exactly one vertex of $X$.
    Then let $I'=I\setminus X$, and $H'=H\cap X$.
    Also, let $M'$ be the set of edges of $M$
    having their endpoints on $I'\cup H'$.
    Since $|I|>|H|$, $I$ is not fully contained in $X$,
    and thus $I'\neq\emptyset$.

    Then we show that $(I',H', M')$ is a crown for the triangles of $G$.
    First, $I'$ does not violate the condition in Definition~\ref{def:crown} since $I'\subset I$.
    Then consider the condition for $H'$.
    Note that $H'$ is contained in $X$, and $I'$ does not intersect $X$.
    Thus for all triangles $\{x,u,v\}$ with $x\in I'$,
    $uv$ must be in $H$, and then $uv$ must be in $H'$.
    Therefore, $H'$ is the set of all edges of $G'$
    that form triangles together with the vertices of $I'$,
    and thus $H'$ satisfies the secons condition in Definition~\ref{def:crown}.
    Also, by construction, $G_{I',H'}$ is a subgraph of
    $G_{I,H}$ such that all edges of $M'$ are in $G_{I',H'}$,
    and thus $H'$ are matched under $M'$.
\end{proof}

By Lemma~\ref{lem:kernel} and Lemma~\ref{lem:crown},
we can obtain an instance $(G',k')$ equivalent to $(G,k)$
such that the union of the triangles has complexity $O(pk)$
for each instance $(G,k)$ obtained from Section~\ref{sec:tri-branching} in polynomial time.



\begin{theorem}\label{thm:trihit}
Given a disk graph $G$ with its geometric representation, we can find a triangle hitting set of $G$ of size $k$
in $2^{O(k^{2/3}\log k)}n^{O(1)}$ time, if it exists.
Without a geometric representation, we can do this in
$2^{ O(k^{4/5}\log k)}n^{O(1)}$ time.
\end{theorem}
\begin{proof}
    After the branching and kernelization processes,
    we have $2^{O((k/p)\log k)}$ instances
    one of which is a \textsf{YES}-instance such that
    the size of the union of the triangles of $G'$ is at most $O(pk)$.
    For each instance $(G',k')$,
    we remove all vertices not contained in any triangle of $G'$. Then the resulting graph $G'$ has $O(pk)$ vertices.
    Then
    we compute a tree decomposition
    $(T,\beta)$ of $G'$ of weighted treewidth  $O(\sqrt{pk})$ using Lemma~\ref{lem:tw-without-ply}.

    By applying dynamic programming on $(T,\beta)$ as described in Lemma~\ref{lem:dp-trihit}, we can find a triangle hitting set of size $k'$ in $2^{O(\sqrt{pk})}$ time if it exists.
    Therefore, the total running time is $2^{O(\sqrt{pk})} \cdot 2^{ O((k/p)\log k)}$. By letting $p=k^{1/3}$, we have $2^{ O(k^{2/3}\log k)}n^{O(1)}$-time.
    
    If we are not given a geometric representation,
    we cannot use a tree decomposition of bounded weighted width.
    Instead, we can compute a tree decomposition
    $(T,\beta)$ of $G'$ of treewidth  $O(p\sqrt{pk})$.
    Then the total running time is $2^{O(p^{3/2}\sqrt{k})} \cdot 2^{ O((k/p)\log k)}$. By letting $p=k^{1/5}$, we have 
    $2^{O(k^{4/5}\log k)}n^{O(1)}$-time.
\end{proof}

\section{Feedback Vertex Set and Odd Cycle Transversal}\label{sec:fvs}
In this section, we show that
the algorithms in~\cite{doi:10.1137/1.9781611977073.80}  for \fvs{} and \oct{} indeed take 
$2^{\tilde O(k^{9/10})}n^{O(1)}$ time and $2^{\tilde O(k^{19/20})}n^{O(1)}$ time algorithms respectively.
It is shown in~\cite{doi:10.1137/1.9781611977073.80} that they take $2^{O(k^{13/14})}n^{O(1)}$ time and $2^{O(k^{27/28})}n^{O(1)}$ time, respectively, but we give a better analysis. 
We can obtain non-robust algorithms for these problems with better running times, but we omit the description of them.

We obtained better time bounds by 
classifying the vertices in a kernel used in~\cite{doi:10.1137/1.9781611977073.80} into two types, and 
by using the higher order additively weighted Voronoi diagrams. 
To make the paper self-contained, we present the algorithms in~\cite{doi:10.1137/1.9781611977073.80} here. 
The following lemma is a main observation of~\cite{doi:10.1137/1.9781611977073.80}. 
Indeed, the statement of the lemma given by~\cite{doi:10.1137/1.9781611977073.80} is stronger than this, but
the following statement is sufficient for our purpose. 
We say a vertex $v$ is \emph{deep} for a subset $F$ of $V$ if 
all neighbors of $v$ in $G$ are contained in $F$. 
For a subset $Q$ of $V(G)$, let $G/Q$ be the graph obtained 
from $G$ by contracting each connected component of $G[Q]$ into a single vertex.
\begin{restatable}[Theorem 1.2 of \cite{doi:10.1137/1.9781611977073.80}]{lemma}{previous}\label{lem:theorem2}
    Let $G$ be a disk graph that has a realization of ply $p$.
    For a subset $F$ of $V$, let $F^*$ be the union of $F$ and the set of all deep vertices for $F$.
    If $F$ contains a core of $G$, then the arrangement graph of $G/Q$ has treewidth 
    $O(\max\{\sqrt{|F^*|\cdot w} \cdot  p^{1.5}, w \})$ for a set $Q\subseteq V\setminus F^*$, where 
    $w$ is the treewidth of $(G-F^*)/Q$. 
\end{restatable}
In Section~\ref{sec:theorem1}, we show that the size of $F^*$ is $O(p^2|F|)$. 
This is our main contribution in this section. 
Then we solve the two cycle hitting problems using dynamic programming on a tree decomposition of bounded treewidth.  

\subsubsection{Branching and Cleaning Process.} We first apply the two-step branching process as we did for \trihit.
Note that if a vertex set $F$ is a feedback vertex set or odd cycle transversal, then $F$ is a triangle hitting set.

We apply the cleaning process for
each instance $(G,k)$ we obtained from
the earlier branching process.
First, we remove all vertices of degree one.
Then we keep $O(1)$ vertices from each class of \emph{false twins}  as follows.
A set of vertices is called a \emph{false twin} if
they are pairwise non-adjacent, and they have the same neighborhood.
As observed in~\cite{doi:10.1137/1.9781611977073.80}, for each class of false twins, every minimal feedback vertex set either contains the entire class except for at most one vertex, or none of the vertices in that class.  
Thus 
we keep only one vertex (an arbitrary one) from each class of \emph{false twins} if $(G,k)$ is an instance of \fvs{}. 
Similarly, for each class of false twins, every
odd cycle transversal
contains the entire class except for at most two vertices, or none of the vertices in that class.
Therefore, 
when we deal with \oct, we keep only two vertices from each class of \emph{false twins}. 
We remember how many vertices each kept vertex represents and make use of this information in DP.

We have $2^{O((k/p)\log k)}$ instances one of which is a  \textsf{YES}-instance $(G,k)$ where $G$ has a core of size $O(pk)$ and a geometric representation of ply $p$. Moreover, at most two vertices are in each class of false twins.

\subsection{The Number of Deep Vertices}
\label{sec:theorem1}
Let $(G,k)$ be an \textsf{YES}-instance obtained from the branching and cleaning process.
Let $F$ be a subset of $V(G)$ containing a core of $G$.
The disks of $G-F$ have ply at most two. 
In this section, we give an upper bound on the number of \emph{deep vertices} for $F$.
For this, we classify the vertices of $V-F$ in two types: regular and irregular vertices. 
See Fig.~\ref{fig:irregular}(a). 

\begin{definition}\label{obs:irregualr}
 A vertex $v$ of $V-F$ is said to be \emph{irregular} if $v$ belongs to one of the following types:
    \begin{itemize}\label{def:irregular}
        \item $D(v)$ contains a vertex of $\mathcal A$, 
        \item $D(v)$ is contained in a face of $\mathcal A$, 
        \item $D(v)$ contains a disk represented by a vertex of $F$, or
        \item $D(v)$ intersects three edges of $\mathcal A$ incident to the same face of $\mathcal A$,  
    \end{itemize}
where $\mathcal A$ denotes the arrangement of the disks represented by $F$. If $v$ does not belong to any of the types, we say $v$ is  regular. 
\end{definition}

\begin{figure}[t]
    \centering
    \includegraphics[width=0.55\textwidth]{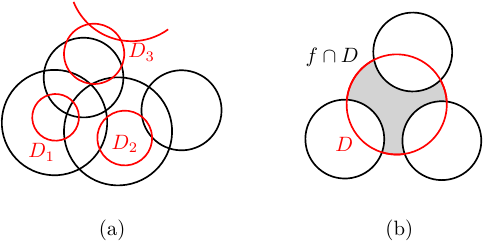}
    \caption{\small (a) The vertices of $F$ are colored black, and the vertices of $V-F$ are colored red. $D_1$ is deep and irregular, $D_2$ is deep and regular, and $D_3$ is shallow and irregular. (b) $D$ is irregular due to the fourth case of Definition~\ref{def:irregular}. 
    }
    \label{fig:irregular} 
\end{figure}

\begin{lemma}\label{lem:regular}
    The number of deep and regular vertices is $O(p^2|F|)$. 
\end{lemma}
\begin{proof}
     If $v$ is regular, the neighbors of $v$ in $G$ form two cliques.
    To see this, observe that 
    the part of $\mathcal A$ restricted to $D(v)$ consists of \emph{parallel} arcs.
    That is, the arcs can be sorted in a way that any two consecutive arcs come from the same face of $\mathcal A$. 
    Also, any two arcs which are not consecutive in the sorted list 
    are not incident to a common face of $\mathcal A$. In this case, 
    there are two points $x$ and $y$
    on the boundary of $D(v)$ such that 
    the line segment $xy$ intersects all disks represented by $F$
    intersecting $D(v)$. Then the disks of $F$ intersecting $D(v)$ contains either $x$ or $y$.
    Therefore, the neighbors of $v$ form at most two cliques.
    
    Since each clique has size at most $p$,
    a deep and regular vertex has at most $2p$ neighbors in $G$. 
    Let $v$ be a deep and regular vertex 
    having exactly $r$ neighbors. 
    Consider the additively weighted order-$r$ Voronoi diagram $r\textsf{VD}$ of 
    $F$. Here, the distance between a disk $D'$ of $F$
    and a point $x$ in the plane is defined as $d(c,x)-w$,
    where $c$ is the center of $D'$ and $w$ is the radius of $D'$. 
    Then the center of $D(v)$ is contained in a Voronoi  region of the order-$r$ Voronoi diagram whose sites
    are exactly the neighbors of $v$. Since no two deep vertices have the same neighborhood in $F$,
    each Voronoi region contains at most one deep vertex. 
    Therefore, the number of deep and regular vertices having exactly $r$ neighbors is linear in the complexity of $r\textsf{VD}$ for $r\leq 2p$. 
    By Lemma~\ref{lem:complexity}, its complexity is $O(r|F|)$.
    Thus, 
    the number of deep and regular vertices of $G$ is $O(p^2|F|)$.
\end{proof}

For irregular vertices, we make use of the fact that the ply of the disks represented by those vertices
is at most two. It is not difficult to see that the number of irregular vertices of the first three types is
$O(p|F|)$. This is because the complexity of $\mathcal A$ is $O(p|F|)$. 
To analyze the number of irregular vertices of the fourth type, we use the following lemma. 

\begin{figure}[t]
    \centering
    \includegraphics[width=0.7\textwidth]{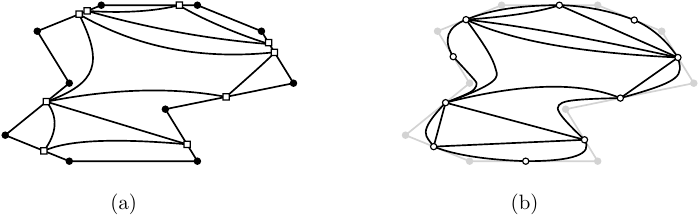}
    \caption{\small (a) $S$ is the polygon with vertices marked with black disks, and each region of $\mathcal F$ has vertices marked with white boxes. (b) The vertices of $H$ are marked with white.}
    \label{fig:planar-subdivision}
\end{figure}
\begin{lemma}\label{lem:planar-subdivision}
    Let $S$ be a connected region with $|S|$ edges in the plane, and $\mathcal F$ be a set of interior-disjoint regions contained in $S$ such that 
    each region intersects at least three edges of $S$. 
    Then the number of regions of $\mathcal F$ is $O(|S|)$. 
    Also, the total number of vertices of the regions of $\mathcal F$ is $O(|S|)$. 
\end{lemma}
\begin{proof}
    We consider the following planar graph $H$: each vertex of $H$ corresponds to an edge of $S$.
    Also, for each region of $\mathcal F$ with vertices $x_1,x_2,\ldots, x_t$ in order, we add edges 
    connecting the vertices corresponding to the edges of $\mathcal S$ containing $x_i$ and $x_{i+1}$ for all indices $1\leq i<t$. 
    In addition to this, for each vertex $v$ of $S$, we add an edge connecting the vertices of $H$ corresponding
    the two edges of $S$ incident to $v$. 
    See Fig.~\ref{fig:planar-subdivision}. 
    There might be parallel edges in $H$, but a pair of vertices has at most two parallel edges.
    Let $F_1$ be the set of faces bounded by a pair of parallel edges, and let $F_2$ be the set of the other faces of $H$.
    Then each region of $\mathcal F$ corresponds to a face of $F_2$. 
    To prove the first part of the lemma, 
    it suffices to show that the size of $F_2$ is linear in the number of edges of $ S$.
    
    By the Euler's formula, $|E_H| \leq |V_H|+|F_1|+|F_2|-2$, where $V_H$ and $E_H$ denote the set of vertices and
    edges of $H$, respectively. Note that no vertex of $H$ has degree at most one. 
    Since every face of $F_2$ is incident to at least three edges of $H$, and every edge of $E_H$ is incident to exactly
    two faces of $H$, $2|E_H| \geq 3|F_2|+2|F_1|$. 
    Thus we have $2|F_1|+3|F_2| \leq 2|V_H| + 2|F_1| + 2|F_2| -4$. 
    That is, $|F_2| \leq 2|V_H|-4$.
    This implies the first part of the lemma. 
    
    Now consider the second part of the lemma. Since a pair of vertices has at most two parallel edges,
    $|F_1| \leq |E_H|/2$. Therefore, $|E_H| \leq 2|V_H|+2|F_2|-4 = O(|V_H|)$. 
    Since every edge of the regions of $\mathcal F$ corresponds to an edge of $H$, 
    the total number of edges of the regions of $\mathcal F$ is $O(| S|)$, and thus 
    the total number of vertices of them is $O(|S|)$. 
\end{proof}

\begin{lemma}\label{lem:additive}
    The number of irregular vertices of $V$
    is $O(p|F|)$. 
\end{lemma}
\begin{proof}
 Since the complexity of $\mathcal A$ is $O(p|F|)$, the number of vertices of $G$ 
    belonging to the first three cases is $O(p|F|)$. 
    For the fourth case in Definition~\ref{def:irregular}, let $f$ be a face of $\mathcal A$ such that $D\cap f$ has at least three circular arcs 
    which comes from the boundary of $f$. 
    In this case,  we replace $D$ with $D\cap f$. 
    Since $D\cap f \subseteq D$, all resulting regions 
    have ply two. See Fig.~\ref{fig:irregular}(b). 
    For a face $f$ having $|f|$ edges, 
    there are $O(|f|)$ such regions by Lemma~\ref{lem:planar-subdivision}. Therefore,
    the number of vertices belonging to the fourth case
    is linear in the complexity of $\mathcal A$, which is  $O(p|F|)$.
    Therefore, the lemma holds.
\end{proof}

\subsection{Feedback Vertex Set}\label{sec:fvs-tw}
In this section, we show how to compute a feedback vertex set of size $k$ in 
$2^{\tilde O(k^{7/8})}n^{O(1)}$ time, if it exists. 
Without a geometric representation of $G$, we can do this in
$2^{\tilde O(k^{9/10})}n^{O(1)}$ time.
Let $(G,k)$ be an \textsf{YES}-instance of \fvs{} from the branching and cleaning processes.
Then we have a core $W$ of $G$ of size $O(pk)$. Let $F$ be the union of a feedback vertex set of size $2k$ 
and the vertex set of all triangles of core $W$ of $G$, and let $F^*$ be the union of $F$ and all deep vertices for $F$. The size of $F^*$ is $O(p^3k)$ by Lemmas~\ref{lem:regular} and~\ref{lem:additive}.
Since $F^*$ contains a feedback vertex set, $G-F^*$ is a forest. 

Therefore, the treewidth of $G$ is $O(p^{4}\sqrt{k})$ by Lemma~\ref{lem:theorem2}.
Then we compute a tree decomposition of treewidth  $O(p^{4}\sqrt{k})$, and then
apply a dynamic programming on the tree decomposition to solve the problem in $2^{O(p^4\sqrt{k})}$ time. 
The total running time is $2^{O(k^{9/10}\log k)}n^{O(1)}$ by setting $p=k^{1/10}$. 
The weighted treewidth of $G$ is $O(p^3\sqrt k \log p)$ by Lemmas~\ref{lem:tw-arrangment} and~\ref{lem:theorem2}.
If we have a geometric representation of $G$, we compute a tree decomposition
    $(T,\beta)$ of $G'$ of weighted treewidth  $O(p^{3}\sqrt{k}\log p)$ using Lemma~\ref{lem:tw-without-ply}.
By applying dynamic programming on $(T,\beta)$ as described in Lemma~\ref{lem:dp-fvs}, we can compute a feedback vertex set of size $k$ in $2^{O(p^{3}\sqrt{k}\log p)}$ time.
Note that the dynamic programming on \fvs{} requires global connectivity information but we can track the connectivity by rank-based approach. See also \cite{de2020framework}.
Since we have $2^{O((k/p)\log k)}$ instances from branching process, the total running time is $2^{O(k^{7/8}\log k)}n^{O(1)}$
by setting $p=k^{1/8}$.

\begin{theorem}
Given a disk graph $G$, we can find a feedback vertex set of size $k$
in $2^{O(k^{9/10}\log k)}n^{O(1)}$ time, if it exists.
Given a disk graph $G$ together with its geometric representation, we can do this in 
in $2^{O(k^{7/8}\log k)}n^{O(1)}$ time.
\end{theorem}

\subsection{Odd Cycle Transversal} 
In this section, we present an $2^{O(k^{19/20}\log k)}n^{O(1)}$-time randomised algorithm for \oct{} when a geometric representation of $G$ is given. 
In the case of \oct, we do not know how to solve the problem without using a geometric representation of $G$. The algorithm in~\cite{doi:10.1137/1.9781611977073.80} also uses a geometric representation in this case.

Let $(G,k)$ be an \textsf{YES}-instance of \oct{} obtained from the
branching and cleaning processes. We have a core $W$ of $G$ of size $O(pk)$.
Let $F$ be the union of the triangles of $W$.  
Furthermore, let $F^*$ be the union of $F$ and the set of all deep  vertices for $F$ of size $O(p^3k)$ by Lemmas~\ref{lem:regular} and~\ref{lem:additive}. 
Let $G^*=G-F^*$. Since $G^*$ does not contain a triangle, it is planar. 
Thus we use a contraction-decomposition theorem on $G^*$ of~\cite{bandyapadhyay2022subexponential}.

\begin{lemma}[Lemma 1.1 of \cite{bandyapadhyay2022subexponential}]\label{lem:theorem12}
    Let $G$ be a planar graph. Then for any $t\in \mathbb N$, there exist disjoint sets  $Z_1,\ldots, Z_p\subseteq V(G)$ such that for every $i\in [t]$ and every $Z'\subseteq Z_i$, treewidth of $G/(Z_i\setminus Z')$ is $O(t+|Z'|)$. Moreover, these sets can be computed in polynomial time.
\end{lemma}

In particular, we set $t=\sqrt k$ and compute $t$ sets
$Z_1,\ldots,Z_{t}$ of vertices of $V\setminus F^*$ such that
for any $Z_i$ and any $Z'\subseteq Z_i$, $G^*/(Z_i\setminus Z')$ has treewidth at most $O(\sqrt{k}+|Z'|)$.
Then there is an index $i$ such that for a fixed odd cycle transversal $S$, $S\cap Z_i$ has size $O(\sqrt{k})$. We iterate over every choice of $i$, and every choice of $Z'=S\cap Z_i$ of size at most $O(\sqrt{k})$.
Due to the following lemma,
the number of choices of $Z'$ we have to consider is  $(k^{O(1)})^{\sqrt k}=2^{O(\sqrt k\log k)}$. 

\begin{lemma}[\cite{doi:10.1137/1.9781611977073.80}]
One can compute a candidate set of size $k^{O(1)}$ for a solution of the \oct{} problem in polynomial time with success probability at least $1-1/2^n$.
\end{lemma}


For each iteration, we contract $Z_i\setminus Z'$ and then
apply a dynamic programming on a tree decomposition of $G/(Z_i\setminus Z')$.
The number of iterations is $2^{O(\sqrt{k}\log k)}$, and $G^*/(Z_i\setminus Z')$ has treewidth $O(\sqrt k)$. 
By Lemma~\ref{lem:tw-arrangment}, the weighted treewidth and the treewidth of $G/(Z_i\setminus Z')$ are $O(p^3k^{3/4}\log p)$ and $O(p^4k^{3/4})$, respectively.
We obtain the desired running times by setting $p=k^{1/16}$ (for a robust algorithm) and $p=k^{1/20}$ (for a non-robust algorithm).
Then Lemma~\ref{lem:dp-oct} implies the following theorem.

\begin{theorem}
Given a disk graph $G$, we can find a odd cycle transversal of size $k$
in $2^{O(k^{19/20}\log k)}n^{O(1)}$ time \emph{w.h.p}.
Given a disk graph $G$ together with its geometric representation, we can do this in 
in $2^{O(k^{15/16}\log k)}n^{O(1)}$ time \emph{w.h.p}.
\end{theorem}

%
%
%
%

\section{DP on a Tree Decomposition of Bounded Treewidth}\label{sec:dp}
In this section, we give dynamic programming algorithms for \trihit, \fvs, and \oct  
on bounded weighted treewidth graphs. 
Let $(T,\beta)$ be a tree decomposition of $G$ with weighted width $O(w)$. 
Recall that $\beta(t)$ has a clique partition of weight $O(w)$. 
Let $\mathcal C(t)$ be the set of cliques in the clique partition of $\beta(t)$, and 
let $\mathcal D(t)$ be the set of disks represented by the vertices of $\beta(t)$.
Also, we slightly abuse the notation so that
    $\mathcal{D}(C)$ denotes the set of disks represented by the vertices of a clique $C$.
For the set $\mathcal V\subset \mathcal D$ of disks, we again abuse the notation so that $G[\mathcal V]$ denotes the subgraph of $G$ induced by the set of vertices whose corresponding disk is contained in $\mathcal V$.
A tree decomposition $(T,\beta)$ is \emph{nice} if $T$ is a rooted tree and each node $t$ of $T$ belongs to one of the four categories: 
\begin{itemize}\label{def:nicetd}
        \item \textbf{Leaf node.} $\beta(t)$ is empty. 
        \item \textbf{Introduce node.} $t$ has exactly one child $t'$ s.t $\mathcal{C}(t)=\mathcal{C}(t')\cup\{C\}$ for some $C\notin \mathcal{C}(t')$ 
         \item \textbf{Forget node.} $t$ has exactly one child $t'$ s.t $\mathcal{C}(t)=\mathcal{C}(t')\setminus\{C\}$ for some $C\in \mathcal{C}(t')$ 
        \item \textbf{Join node.} $t$ has two children $t_1,t_2$ s.t $\mathcal{C}(t)=\mathcal{C}(t_1)=\mathcal{C}(t_2)$. 
\end{itemize}

It is not difficult to see that, 
given a tree decomposition with weighted width $w$, we can compute a nice tree decomposition with weighted width $O(w)$ in polynomial time.
It is well known that we can convert a tree decomposition of width $w$ into a tree decomposition of width $O(w)$~\cite{pBook}.
A slight modification of this conversion preserves the weighted width and the size of a tree asymptotically. 
Therefore, we may assume that $(T,\beta)$ is a nice tree decomposition of weighted width $O(w)$.
Moreover, the number of nodes of $T$ is $O(wn)$. 
Also, let $\mathcal{V}_t$ be the union of $\mathcal{D}(t')$ for all nodes $t'$ in the subtree of $T$ rooted at $t$. 

\begin{lemma}\label{lem:dp-trihit}
Given a tree decomposition $(T,\beta)$ of $G$ of weighted treewidth $w$,
we can solve \trihit{} in $2^{O(w)} n$ time. 
\end{lemma}
\begin{proof} 
For each node $t$ and each subset $\mathcal{S}$ of $\mathcal D(t)$, we define the subproblem of \trihit{} as 
\begin{align*}
    c[t,\mathcal{S}]=\begin{cases}\min_{\mathcal{S}\subset \mathcal{S}'\subset \mathcal{V}_t} |\mathcal{S}'|  & \text{     s.t }  G[\mathcal{V}_t\setminus \mathcal S'] \text{  is triangle-free}\\\infty & \text{  if there is no such } \mathcal{S}'
    \end{cases}
\end{align*}
    A solution of \trihit{} must contain all but two vertices from each clique. 
    Therefore, since 
    \begin{align*}
        \prod_{C_i\in \mathcal C} |C_i|^2=\exp \left( \sum_{C_i\in \mathcal C} 2\log |C_i| \right) =2^{O(w)},
    \end{align*}
    for each node $t$, all but $2^{O(w)}$ 
    subproblems are set to infinity. 
    We say that a subproblem $c[t, \mathcal S]$ \emph{has} a solution if it has a finite value. Moreover, we say $\mathcal S'$ is a \emph{partial solution} of $c[t,\mathcal S]$ if $\mathcal S\subset \mathcal S'\subset \mathcal V_t$ and $G[\mathcal V_t\setminus \mathcal S']$ is triangle-free. We say that a partial solution is \emph{optimal} for $c[t,\mathcal S]$ if its size is minimum among all partial solutions.
    We give formulas for each of the four cases, and then apply these formulas in a bottom-up manner on $T$.
    Eventually, we compute $c[r,\emptyset]$, where $r$ is the root of $T$.
    
    \paragraph{\textbf{Leaf node.}} If $t$ is a leaf, we have only one subproblem $c[t,\emptyset]=0$. 
    \paragraph{\textbf{Introduce node.}} Let $t$ has one child $t'$ with $\mathcal{C}(t)=\mathcal{C}(t')\cup\{C\}$. We claim that the following formula holds: 
    \begin{align*}
    c[t,\mathcal{S}]= \begin{cases}c[t',\mathcal{S}\setminus \mathcal{D}(C)]+|\mathcal{S}\cap \mathcal{D}(C)| & \text{if } G[\mathcal{D}(t)\setminus \mathcal S] \text{ is triangle-free}\\ \infty & \text{otherwise.} 
    \end{cases}
    \end{align*}
    Consider a possible optimal partial solution $\mathcal{S}'$ for the subproblem $c[t,\mathcal{S}]$. 
    Then $\mathcal{S}'\cap 
    \mathcal{D}(t')$ is exactly $\mathcal{S}\setminus \mathcal{D}(C)$, and $\mathcal{S}'\setminus \mathcal D(C)$ is also a partial solution of $c[t',\mathcal{S}\setminus \mathcal{D}(C)]$. To see the $\mathcal{S}'\setminus \mathcal D(C)$ is optimal for $c[t',\mathcal{S}\setminus \mathcal{D}(C)]$, suppose that there is an optimal partial solution $\mathcal{S}''$ for $c[t',\mathcal{S}\setminus \mathcal{D}(C)]$ such that $\mathcal{S}''\cap \mathcal D(t')=\mathcal{S}\setminus \mathcal{D}(C)$. 
    Then $\mathcal{S}''\cup \mathcal{S}$ is also a solution of $c[t, \mathcal S]$ because the disk $D$ in $\mathcal D(C)$ dose not intersect to a disk in $\mathcal V_t\setminus \mathcal D(t)$. Therefore, $\mathcal{S}'\setminus \mathcal D(C)$ is an optimal partial solution of $c[t', \mathcal S\setminus \mathcal D(C)]$, so the formula holds. 

    \paragraph{\textbf{Forget node.}} Let $t$ has one child $t'$ with $\mathcal{C}(t)=\mathcal{C}(t')\setminus\{C\}$. We claim that the following formula holds:
    \begin{align*}
    c[t,\mathcal{S}]= \min_{\mathcal{A}\subset \mathcal{D}(C)} c[t',\mathcal{S}\cup \mathcal{A}] 
    \end{align*}
    For an optimal partial solution $\mathcal{S}'$ for $c[t,\mathcal S]$, we set $\mathcal{A}=\mathcal{S}'\cap \mathcal{D}(C)$. Since $\mathcal{V}_t=\mathcal{V}_{t'}$, $\mathcal{S}'$ is a partial solution of $c[t', \mathcal S\cup \mathcal A]$.
    Conversely, suppose $c[t',\mathcal{S}\cup \mathcal{A}]$ is minimum among all $\mathcal A\subset \mathcal D(C)$.
    Let $\mathcal S'$ be an optimal partial solution of $c[t',\mathcal S\cup \mathcal A]$. 
    Then $\mathcal{S}'$ is also a partial solution of $c[t,\mathcal C]$ because 
    $\mathcal V_{t_1}=\mathcal V_t$ and $\mathcal S'\cap \mathcal D(t)=\mathcal S$.
    
    \paragraph{\textbf{Join node.}} Let $t$ has two children $t_1$ and $t_2$. We claim that 
    \begin{align*}
        c[t,\mathcal{S}]=c[t_1,\mathcal{S}]+c[t_2,\mathcal{S}]-|\mathcal{S}|.
    \end{align*}
    Let $\mathcal S'$ is an optimal partial solution of $c[t,\mathcal S]$. Then  $\mathcal S\cap \mathcal V_{t_1}$ (and $\mathcal S\cap \mathcal V_{t_2}$) is a partial solution of $c[t_1, \mathcal S]$ (and $c[t_2,\mathcal S]$).
    Conversely, if there are optimal partial solutions $\mathcal S_1$ of $c[t_1,\mathcal S]$ and $\mathcal S_2$ of $c[t_2, \mathcal S]$, we show that $\mathcal S_1\cup \mathcal S_2$ is a partial solution of $c[t,\mathcal S]$.
    Suppose
     there are three disks $D_x,D_y,D_z\in \mathcal{V}_{t_1}\cup \mathcal{V}_{t_2}$ that pairwise intersect. 
    We may assume that $D_x,D_y\in \mathcal{V}_{t_1}$. If $D_z\in \mathcal{V}_{t_1}$, 
    an induced subgraph $G[\{D_x,D_y,D_z\}]$ of $G[\mathcal V_{t_1}]$ is a triangle.
    Then the partial solution $\mathcal S_1$ of $c[t_1, \mathcal S]$ must hit one of the three disks $D_x,D_y$ and $D_z$. 
    Otherwise, both $D_x$ and $D_y$ intersect with $D_z\in \mathcal{V}_{t_2}\setminus \mathcal{V}_{t_1}$, which implies that $D_x,D_y\in \mathcal{D}(t_2)$. Then, the partial solution $\mathcal S_2$ of $c[t_2, \mathcal S]$ hits the triangle. For both cases, $\mathcal S_1\cup \mathcal S_2$ hits the set of three disks. Therefore,   $\mathcal S_1\cup \mathcal S_2$ is a partial solution of $c[t, \mathcal S]$.

    Now we compute all finite $c[t,\mathcal{S}]$ of node $t$ in $2^{O(w)}$ time.
    The total time complexity is $2^{O(w)}n$ because the number of nodes in tree is $O(wn)$.  
\end{proof}

\begin{lemma}\label{lem:dp-fvs}
Given a tree decomposition $(T,\beta)$ of $G$ of weighted treewidth $w$,
we can solve \fvs{} in $2^{O(w)} n$ time. 
\end{lemma}
\begin{proof}(Sketch.)
A dynamic programming for \fvs{} is more involved because we need to maintain global connectivity information. 
In particular, for each node $t$, each subset $\mathcal S$ of $\mathcal D(t)$ and each partition $\mathcal P$ of $\mathcal D(t)\setminus \mathcal S$, we define the subproblem of \fvs{} as 
\begin{align}
    c[t,\mathcal S, \mathcal P]=
    & \min_{S\subset S'\subset \mathcal V_t} |\mathcal S'| \text{ s.t } G[\mathcal V_t\setminus \mathcal S'] \text{ is cycle-free, and}  \\
    & \text{for each } P\in \mathcal P, \text{ the vertices corresponding to disks in } P \\
    & \text{are contained in a same } \text{connected component of} G[\mathcal V_t\setminus S']
\end{align}

We say that the subproblem \emph{has} a solution if such a superset $\mathcal S'$ of $\mathcal S$ exists. Then we say $\mathcal S'$ is a \emph{partial solution} of $c[t,\mathcal S, \mathcal P]$.
Similar to Lemma~\ref{lem:dp-trihit}, we have $2^{O(w)}$ choices of $\mathcal S$ that 
$c[t,\mathcal S, \cdot]$ has a solution.

Intuitively, if $\mathcal S'$ is a partial solution and $D_1,D_2$ are contained in the same partition class of $\mathcal P$, then two vertices of $G$  correspond to $D_1$ and $D_2$ are in the same connected component of $G[\mathcal V_t\setminus S']$. However, the number of choices of $\mathcal P$ is $w^{O(w)}$ because $|\mathcal V_t\setminus \mathcal S'|=O(w)$. 
To reduce the number of subproblems we have to consider, we use a rank-based approach~\cite{bodlaender2015rank}. 
The rank-based approach ensures that for the fixed $\mathcal S$, there are $2^{O(w)}$ partitions among all possible $w^{O(w)}$ choices of $\mathcal P$ so that they cover all possibilities of the connected components of $G[\mathcal V_t\setminus \mathcal S']$.

One small issue is that the 
rank-based approach is designed to get maximum connectivity whereas in \fvs{} we are aim to minimize the connectivity. As did in \cite{bodlaender2015rank}, we add a special universal vertex $v_0$ to the graph $G[\mathcal V_t]$ and increase the weighted width by 1. Now the task is to determine if there is a connected subgraph of $G[\mathcal V_t]$ that contains $v_0$ after deleting a partial solution $\mathcal S'$ from $G[\mathcal V_t]$. 
In particular, the task is to compute maximum connectivity among graphs of $k=O(w)$ vertices and $k-1$ edges. 
Now we use the rank-based approach and we compute $2^{O(w)}$ partitions $\mathcal P$ for each node $t$ and each subset $\mathcal S\subset \mathcal D(t)$. 
This completes the proof.
\end{proof}

\begin{lemma}\label{lem:dp-oct}
Given a tree decomposition $(T,\beta)$ of $G$ of weighted treewidth $w$,
we can solve \oct{} in $2^{O(w)} n$ time. 
\end{lemma}
\begin{proof}(Sketch.) 
For each node $t$ and each function $f:\mathcal{D}(t)\rightarrow \{0,1,2\}$, we define the subproblem of \oct{} as 
    \begin{align*}
        c[t,f]=
        & \min_{g} |g^{-1}(0)| \text{  where  } g:\mathcal{V}_t\rightarrow \{0,1,2\} \text{  s.t} \\
                & g(x)=f(x) \text{ if } x\in \mathcal{D}(t) \And \forall (x,y)\in E(G[\mathcal{V}_t]), \{g(x), g(y)\}\neq \{1,2\}
    \end{align*}
    Similar to Lemma~\ref{lem:dp-trihit}, we say that $g:\mathcal{V}_t\rightarrow \{0,1,2\}$ is a \emph{partial solution} of $c[t,f]$ if $g$ satisfies the conditions of the definition above. Also, we say $g$ is \emph{optimal} if the size of $|g^{-1}(0)|$ is minimum over all partial solutions of $g[t,f]$.
    Intuitively, as a solution $g$ of a subproblem $c[t,f]$, we remember not only the set of disks to be deleted as in Lemma~\ref{lem:dp-trihit} but also the bipartition of $\mathcal{V}_t\setminus f^{-1}(0)$. The number of subproblems  that have a solution is again $2^{O(w)}$ because for each clique $C\in \mathcal C(t)$, all but two vertices should be mapped to zero. 
    We focus on the join nodes because the formulas of other categories can be obtained in a straightforward way. Similar to \trihit{}, we claim that the following formula holds: 
    \begin{align*}
        c[t,f]=c[t_1,f]+c[t_2,f]-|f^{-1}(0)| 
    \end{align*}
    The disk $D_1$ in $\mathcal{V}_{t_1}\setminus \mathcal{V}_{t_2}$ and the disk $D_2$ in $\mathcal{V}_{t_2}\setminus \mathcal{V}_{t_1}$ do not intersect by definition. 
    Therefore, if $g$ is an optimal partial solution of $c[t,f]$, the domain restriction $g_1$ (and $g_2$) of $g$ with respect to  $\mathcal{V}_{t_1}$ (and $\mathcal{V}_{t_2}$) is a partial solution of $c[t_1,f]$ (and $c[t_2,f]$). Conversely, the domain union of two optimal partial solutions $g_1$ (and $g_2$) of $c[t_1,f]$ (and $c[t_2,f]$) makes a partial solution of $\mathcal{V}_t$ because $g_1(x)=g_2(x)=f(x)$ for all $x\in \mathcal D(v)$. This shows that the formula holds. 
    Then \oct{} can be solved in $2^{O(w)}n^{O(1)}$ time.
\end{proof}

\bibliography{papers}
 \bibliographystyle{splncs04}
\appendix


\end{document}